\documentclass[12pt,aps,prd]{revtex4}
\usepackage{amsmath,amssymb,amsthm}
\usepackage[makeroom]{cancel}

\newtheorem{theorem}{Theorem}[section]

\newtheorem{proposition}[theorem]{Proposition}
\newtheorem{corollary}[theorem]{Corollary}
\newtheorem{definition}{Definition}

\newenvironment{remark}[1][Remark]{\begin{trivlist}
\item[\hskip \labelsep {\bfseries #1}]}{\end{trivlist}}

\begin{document}

\title{Bracket-Preserving property of Anchor Maps and Applications to Various Brackets}
\author{S. Srinivas Rau}
\email{srinivasrau@ifheindia.org}
\affiliation{Department of Mathematics, Faculty of Science and Technology, The ICFAI Foundation for Higher Education, Dontanapally, Hyderabad, 501 203.}
\author{T. Shreecharan}
\email{shreecharan@ifheindia.org}
\affiliation{Department of Physics, Faculty of Science and Technology, The ICFAI Foundation for Higher Education, Dontanapally, Hyderabad, 501 203.}
\begin{abstract}
\begin{center}
\small{Abstract}
\end{center}
Let $E \rightarrow M$ be a smooth vector bundle with a bilinear product on $\Gamma(E)$ satisfying the Jacobi identity. Assuming only the existence of an anchor map $\mathfrak{a}$ we show that $\mathfrak{a}([X,Y]) = [\mathfrak{a}X,\mathfrak{a}Y]_c$. This gives the redundancy of the homomorphism condition in the definition of Leibniz algebroid (in particular if it arises from a Nambu-Poisson manifold); an aspect not addressed in the literature. We apply our result to the brackets of Hagiwara, Ibanez et. al; we settle an old query of Uchino on redundancy for Courant bracket.
\end{abstract}
\maketitle

\section{Introduction}

The sucessful geometrisation of classical mechanics \cite{Arnold,Abraham} has encouraged researchers to formulate mechanics of more generality than that of Lagrangian and Hamiltonian mechanics. This  has had an impact on the understanding of physical systems that were previously beyond the scope of traditional mechanics. One of the earliest attempts of generalisation of mechanics has been via the use of Lie algebroids \cite{Weinstein,Liberman}. This formulation was further extended to that of mechanics on almost Lie algebroids and almost Leibniz algebroids. This was achieved by dropping the Jacobi identity in the definition of Lie algebroid leading to almost Lie algebroid \cite{Marrero} and further dropping the requirement of skew-symmetry leading to almost Leibniz algebroid structure \cite{Balseiro}. Another line of construction of mechanics is the one advocated in Ortega and Planas-Bielsa \cite{Ortega} which is based on the Leibniz bracket as introduced by Grabowski and Urbanski \cite{Grabowski2}. It must be mentioned here that this strucure is different from the concept of Leibniz algebroid that is a natural generalisation of a Lie algebroid obtained by discarding the skew-symmetric condition as introduced by Loday \cite{revgrab}.

Recently Lie algebroids have been used to formulate more general gauge theories than Yang-Mills \cite{Strobl2}. This approach has yielded rich dividends. Poisson sigma model \cite{Schaller,Ikeda} a prototype of a Lie algebroid gauge theory has provided a field theoretic insight into the deformation quantisation scheme of Kontsevich. It has also shown some promising glimpses of uniting gravity and gauge theory in a common framework \cite{Strobl1} atleast only in two dimensions as of now.

Our contribution is to give a simple computational proof that a commonly used axiom for Lie/Leibniz/Courant algebroids is redundant: if an anchor map exists and if the bracket satisfies the Jacobi/Leibniz identity then the anchor preserves brackets. This perhaps simplifies the task for the physicist since many proofs involve the bracket-preserving condition. In view of our proposition (\ref{prop1}), the bracket preserving condition is equivalent to the Jacobi identity for Leibniz and Lie algebroids. For Courant algebroids, Prop. (\ref{prop1}) answers a question of Uchino \cite{Uchino}.

We point out that the Hagiwara and Courant brackets arise from Dirac's theory of constraints. Our treatment below avoids the heavy background of the original authors. An interesting new bracket on $(p-1)$ forms is introduced and studied.

\section{Bracket preserving property of anchor}

\noindent Suppose $E \rightarrow M$ is a smooth vector bundle. Let $[\, , \,]$ be a bilinear bracket on the vector space of smooth sections $\Gamma(E)$. Note that $\Gamma(E)$ is a faithful module over the ring $\mathcal{C}^\infty(M)$. \textit{We assume}
\begin{enumerate}

\item $[X,[Y,Z]] = [[X,Y],Z] + [[Y,[X,Z]] \quad \forall \quad X,Y,Z \in \Gamma(E)$.

\item Let $\mathcal{T}(\mathcal{C}^\infty(M), \mathcal{C}^\infty(M))$ be the set of transformations (self-mappings) of $\mathcal{C}^\infty(M)$. Suppose there is a map $\mathfrak{a}: \Gamma(E) \rightarrow \mathcal{T}$ such that $\forall \, f \in \mathcal{C}^\infty(M)$ and $X, Y \in \Gamma(E)$ one has \\
    $(\mathfrak{a}(X)f)Y = [X,fY] - f[X,Y]$ \\
Since $\Gamma(E)$ is faithful, $(\mathfrak{a}(X)f)Y = 0 \quad \forall \quad Y$ iff $(\mathfrak{a}(X)f) =0 \, \in \, \mathcal{C}^\infty(M)$

\end{enumerate}

\noindent Then
\begin{proposition} \label{prop1}
\begin{enumerate}

\item $\mathfrak{a}([X,Y]) = [\mathfrak{a}(X),\mathfrak{a}(Y)]_c$ where $[ \, , \,]_c$ is the commutator in $\mathcal{T}: [T, S]_c \, g = T(Sg) - S(Tg)$ $\forall \, g \in \mathcal{C}^\infty(M), \, T,S \, \in \mathcal{T}$.

\item If $\mathfrak{a}$ is a linear map then $\mathfrak{a}([X,Y]) = - \mathfrak{a}([Y,X])$.

\item Each transformation $\mathfrak{a}(X)$ satisfies $\mathfrak{a}(X)(fg) = f(\mathfrak{a}(X)g) + (\mathfrak{a}(X)f)g$ $\forall \, f, g \in \mathcal{C}^\infty(M)$. Thus if $\mathfrak{a}(X)$ is linear then it is a derivation of $\mathcal{C}^\infty(M)$.

\end{enumerate}
\end{proposition}
\begin{proof}
\begin{enumerate}
\item To show that $\mathfrak{a}$ preserves brackets, we choose and fix $f \in \mathcal{C}^\infty(M)$ and $Z \in \Gamma(E)$. Now we claim that $(\mathfrak{a}([X,Y])f)Z = ([\mathfrak{a}(X),\mathfrak{a}(Y)]_c \, f)Z$. By the faithfulness of the module $\Gamma(E)$ we have equality of the brackets for arbitrary $f$. \\
Now the LHS is
\begin{eqnarray*}
(\mathfrak{a}([X,Y])f)Z & = & [[X,Y],fZ] - f [[X,Y],Z] \\
& = & [[X,Y],fZ] - f \big\{[X,[Y,Z]] - [Y,[X,Z]]\big\} \\
& = & [[X,Y],fZ] - f [X,[Y,Z]] + f [Y,[X,Z]]
\end{eqnarray*}
The RHS is
\begin{equation*}
([\mathfrak{a}(X),\mathfrak{a}(Y)]_c \, f)Z = \Big(\mathfrak{a}(X)\big\{\mathfrak{a}(Y)f \big\} - \mathfrak{a}(Y)\big\{\mathfrak{a}(X)f \big\} \Big) Z
\end{equation*}
Setting $\mathfrak{a}(Y)f = g$ and $\mathfrak{a}(X)f = h$ leads to
\begin{eqnarray*}
(\mathfrak{a}([X,Y])f)Z & = & \big(\mathfrak{a}(X)g - \mathfrak{a}(Y)h \big) \\
& = & \big\{[X,gZ] - g[X,Z] \big\} - \{[Y,hZ] - h[Y,Z]\} \\
& = & \big\{[X,\mathfrak{a}(Y)f Z] - \mathfrak{a}(Y)f [X,Z] \big\} - \{[Y,\mathfrak{a}(X)f Z] - \mathfrak{a}(X)f [Y,Z]\} \\
& = & [X,[Y,fZ] -f[Y,Z]] - \big\{[Y,f[X,Z]] - f[Y,[X,Z]] \big\} \\
& - & [Y,[X,f Z]-f[X,Z]] + [X,f[Y,Z]-f[X,[Y,Z]] \\
& = & [X,[Y,fZ]] -\cancel{[X,f[Y,Z]]} - \bcancel{[Y,f[X,Z]]} + f[Y,[X,Z]] \\
& - & [Y,[X,f Z]] + \bcancel{[Y,f[X,Z]]} + \cancel{[X,f[Y,Z]]} - f[X,[Y,Z]] \\
& = & [[X,Y],fZ] + f[Y,[X,Z]] - f[X,[Y,Z]] = \mathrm{LHS} \qed
\end{eqnarray*}
Note that in obtaining the last step we have made use of our assumption (i)

\item If $\mathfrak{a}$ is linear, in particular $\mathfrak{a}(-Z) = - \mathfrak{a}(Z)$ for any $Z$. Therefore
\begin{eqnarray*}
\mathfrak{a}([X,Y]) & = & [\mathfrak{a}(X),\mathfrak{a}(Y)]_c \\
& = & - [\mathfrak{a}(Y),\mathfrak{a}(X)]_c \quad (\mathrm{from \, the \, definition \, of} \, [ \, , \,]_c) \\
& = & - \mathfrak{a}([Y,X]) \qed
\end{eqnarray*}
\item We note that the proof of Skryabin's theorem (Prop. 1.1) \cite{Skryabin}  holds for any $\mathfrak{a}(X)$, which is denoted by $\hat{D}$ by Grabowski (Thm. 1, pg 2) \cite{Grabowski}. This gives the Leibniz property: \\
    $\mathfrak{a}(X)(fg) = f(\mathfrak{a}(X)g + (\mathfrak{a}(X)f)g$ \qed \\
By definition a derivation is a linear map on $\mathcal{C}^\infty(M)$ satisfying the Leibniz property \cite{Koszul}.
\end{enumerate}

\end{proof}

\section{Leibniz Algebroid}

The usual definition of Leibniz algebroid is \cite{Ibanez}:
\begin{definition} \label{leibnizdef}
A Leibniz algebra structure on a real vector space $\mathfrak{g}$ is a $\mathbb{R}$-bilinear map $[[\, , \,]]:\mathfrak{g} \times \mathfrak{g} \rightarrow \mathfrak{g}$ satisfying the Leibniz identity
\begin{equation*}
[[a_1, [[a_2,a_3]] \ ]] - [[ \ [[a_1,a_2]], a_3]] - [[a_2, [[a_1,a_3]] \ ]] = 0 \quad \mathrm{for} \quad a_1,a_2,a_3 \in \mathfrak{g}
\end{equation*}
\end{definition}
\begin{definition} \label{leibalg}
A Leibniz algebroid structure on a differentiable vector bundle $E \rightarrow M$ is a pair that consists of a Leibniz algebra structure $[[\, , \,]]$ on the space $\Gamma(E)$ of the global cross sections of $E \rightarrow M$ and a vector bundle morphism $\varrho: E \rightarrow TM$, called the anchor map, such that the induced map $\varrho: \Gamma(E) \rightarrow \Gamma(TM) = X(M)$ satisfies the following relations:
\begin{enumerate}
\item \textit{$\varrho[[s_1,s_2]] = [\varrho(s_1),\varrho(s_2)]$}
\item \textit{$[[s_1,fs_2]] = f[[s_1,s_2]] + \varrho(s_1)(f)s_2$}
\end{enumerate}
$\forall \, s_1, s_2 \in \Gamma(E)$ and $f \in \mathcal{C}^\infty(M)$. \\
A triple $(E, [[\, , \,]], \varrho)$ is called a Leibniz algebroid over the manifold $M$.
\end{definition}
Leibniz algebroid has been associated with Nambu-Poisson manifold. In fact this association is very interesting since it has been shown that Nambu-Poisson manifold has atleast two different Leibniz algebroid structures. The first one being derived in Ref \cite{Ibanez} and the other in Ref \cite{Hagiwara}. This is not only interesting from a mathematical point of view but physically also it throws up intriguing questions. We reproduce some of the definitions and also the two distinct Leibniz algebroid structures here, for the sake of convenience of the readers.
\begin{definition}
Let $M$ be a smooth $n$-dimensional manifold. A Nambu-Poisson structure on $M$ of order $p$ (with $2\leq p\leq n$) is given by a $p$-vector field which satisfies the fundamental identity.
\end{definition}
\begin{definition}
Let $M$ be a Nambu-Poisson manifold. A Leibniz algebroid attached to $M$ is the triple $(\bigwedge^{p-1}(T^\ast M), [[\, , \,]], \Pi)$, where $[[\, , \,]]: \Omega^{p-1}(M) \times \Omega^{p-1}(M) \rightarrow \Omega^{p-1}(M)$ is the bracket of $(p-1)$ forms, as defined by Iba\~{n}ez et. al. \cite{Ibanez}
\begin{equation} \label{ibala}
[[\alpha,\beta]] = \mathcal{L}_{\Pi(\alpha)} \beta + (-1)^p (\Pi(d\alpha))\beta
\end{equation}
or, as defined by Hagiwara \cite{Hagiwara}
\begin{equation} \label{hagla}
[[\alpha,\beta]] = \mathcal{L}_{\Pi(\alpha)} \beta - \imath_{\Pi(\beta)} d\alpha
\end{equation}
for $\alpha, \beta \in \Omega^{p-1}(M)$ and $\Pi:\bigwedge^{p-1}(T^\ast M) \rightarrow TM$ if the homomorphism of the vector bundles given by $\Pi(\beta)=i(\beta)\Lambda(x)$; $\Lambda$ being the Nambu-Poisson $p$-vector, $\mathcal{L}$ the Lie derivative, and $\imath$ the interior product.
\end{definition}

The redundancy of the homomorphism condition in the definition of Leibniz algebroid can be written as follows, in view of Proposition \ref{prop1} (1)
\begin{corollary}
In the definition of a Leibniz algebroid the bracket-preserving condition (Cond (1) of Def \ref{leibnizdef}) on the anchor map $\varrho$ is redundant. In particular this redundancy holds for Leibniz algebroids arising from Nambu-Poisson manifolds.
\end{corollary}

Though this type of redundancy has been pointed out for Lie algebroids \cite{Grabowski,Marle}, the redundancy for the Leibniz algebroids has not been proved earlier. Apart from the aforementioned result, our derivation gives yet another proof of this redundancy in the Lie algebroid case (the definition includes the skew-symmetric condition for the bracket).

\begin{corollary}
If $[X,Y] = -[Y,X] \quad \forall \, X,Y, \in \Gamma(E)$, (i.e., the bracket is skew-symmetric) and $\mathfrak{a}$ is a linear map then $\mathfrak{a}$ preserves antisymmetry ie $(\mathfrak{a}[X,Y]) = \mathfrak{a}(-[X,Y])$.
\end{corollary}

\section{Analysis of Hagiwara's Bracket}

Let $M$ be a manifold and $\Pi$ a $p$-vector on $M$. Define
\begin{equation}
[[\alpha,\beta]] = \mathcal{L}_{\Pi(\alpha)} \beta - \imath_{\Pi(\beta)} d\alpha
\end{equation}

\begin{proposition}
$\Pi$ is the anchor map for $[[ \ ,\ ]]$
\end{proposition}
\begin{proof}
\begin{eqnarray} \nonumber
[[\alpha,f\beta]] & = & \mathcal{L}_{\Pi(\alpha)}f\beta - \imath_{\Pi(f\beta)} d\alpha \\ \nonumber
& = & f (\mathcal{L}_{\Pi(\alpha)}\beta) + (\mathcal{L}_{\Pi(\alpha)}f)\beta - f \imath_{\Pi(f\beta)} d\alpha \\ \nonumber
& = & f (\mathcal{L}_{\Pi(\alpha)}\beta) - f \imath_{\Pi(f\beta)} d\alpha + (\mathcal{L}_{\Pi(\alpha)}f)\beta  \\ \nonumber
[[\alpha,f\beta]] & = & f [[\alpha,\beta]] + (\mathcal{L}_{\Pi(\alpha)}f)\beta
\end{eqnarray}
\end{proof}
so that $\mathcal{L}_{\Pi(\alpha)}f = \mathfrak{a}f$

\begin{proposition}
$\Pi[[\alpha,\beta]] = [\Pi\alpha,\Pi\beta]_c$ if and only if $\Pi$ is a Nambu-Poisson tensor.
\end{proposition}
\begin{proof}
(See Baraglia \cite{Baraglia} for details). The bracket preserving condition is seen to be equivalent to
\begin{equation}
\mathcal{L}_{\Pi(df_1 \wedge \cdots \wedge df_{p-1})}(\Pi \beta) = \Pi(\mathcal{L}_{\Pi(df_1 \wedge \cdots \wedge df_{p-1})}\beta)
\end{equation}
which in turn is equivalent to the invariance condition
\begin{equation}
\mathcal{L}_{\Pi(df_1 \wedge \cdots \wedge df_{p-1})}(\Pi) = 0
\end{equation}
The last condition is known to be equivalent to the fundamental identity for the Nambu-Poisson bracket is given by
\begin{equation}
\{f_1, \cdots,f_p\} = \Pi(df_1 \wedge \cdots \wedge df_{p-1})
\end{equation}
\end{proof}
\begin{proposition} \label{hagjacobi}
Let $(M,\Pi)$ be a Nambu-Poisson manifold. Then the Leibniz identity holds
\begin{equation*}
[[\alpha, [[\beta,\gamma]] \ ]] - [[ \ [[\alpha,\beta]], \gamma]] - [[\beta, [[\alpha,\gamma]] \ ]] = 0
\end{equation*}
\end{proposition}
\begin{proof}
Let us write down the explicit terms
\begin{eqnarray}
[[\alpha, [[\beta,\gamma]] \ ]] & = & \mathcal{L}_{\Pi(\alpha)} [[\beta,\gamma]] - \imath_{\Pi([[\beta,\gamma]])} d\alpha \\ \nonumber
& = & \mathcal{L}_{\Pi(\alpha)} \big\{ \mathcal{L}_{\Pi(\beta)} \gamma - \imath_{\Pi(\gamma)} d \beta \big \} - \imath_{[\Pi\beta,\Pi\gamma]_c} d\alpha \\ \nonumber
& = & \mathcal{L}_{\Pi(\alpha)} \mathcal{L}_{\Pi(\beta)}\gamma - \mathcal{L}_{\Pi(\alpha)} \imath_{\Pi(\gamma)} d \beta  - \imath_{[\Pi\beta,\Pi\gamma]_c} d\alpha
\end{eqnarray}
\begin{eqnarray}
[[\ [[\alpha,\beta,]]\gamma]] & = & \mathcal{L}_{\Pi([[\alpha,\beta]])} \gamma - \imath_{\Pi(\gamma)} d[[\alpha,\beta]] \\ \nonumber
& = & \mathcal{L}_{[\Pi(\alpha),\Pi(\beta)]} \gamma - \imath_{\Pi(\gamma)} d \big \{ \mathcal{L}_{\Pi(\alpha)} \beta - \imath_{\Pi(\beta)} d\alpha \big\}\\ \nonumber
& = & \mathcal{L}_{[\Pi(\alpha),\Pi(\beta)]_c} \gamma - \imath_{\Pi(\gamma)} \mathcal{L}_{\Pi(\alpha)} d\beta + \imath_{\Pi(\gamma)}d \big(\imath_{\Pi(\beta)} d\alpha \big)
\end{eqnarray}
\begin{eqnarray}
[[\beta, [[\alpha,\gamma]] \ ]] & = & \mathcal{L}_{\Pi(\beta)} [[\alpha,\gamma]] - \imath_{\Pi([[\alpha,\gamma]])} d\beta \\ \nonumber
& = & \mathcal{L}_{\Pi(\beta)} \big\{ \mathcal{L}_{\Pi(\alpha)} \gamma - \imath_{\Pi(\gamma)} d \alpha \big \} - \imath_{[\Pi\alpha,\Pi\gamma]_c} d\beta \\ \nonumber
& = & \mathcal{L}_{\Pi(\beta)} \mathcal{L}_{\Pi(\alpha)}\gamma - \mathcal{L}_{\Pi(\beta)} \imath_{\Pi(\gamma)} d \alpha  - \imath_{[\Pi\alpha,\Pi\gamma]_c} d\beta
\end{eqnarray}
Terms involving Lie derivatives cancel. To show that the rest of the terms are identically equal to zero we use of the following relation
\begin{equation*}
\imath_{[X,Y]_c} \alpha = [\mathcal{L}_X,\imath_Y]_c \ \alpha = [\imath_Y,\mathcal{L}_X]_c \ \alpha
\end{equation*}
and the Cartan formula
\begin{equation*}
\mathcal{L}_X \alpha = \imath_X d \alpha + d \imath_X \alpha
\end{equation*}
\end{proof}
\begin{corollary}
($\bigwedge^{p-1}T^\ast M, [[ \ ,\ ]], \Pi$) is a Leibniz algebroid.
\end{corollary}
\begin{corollary} \label{cirHag}
For the bracket $[[ \ ,\ ]]$ the following are equivalent :
\begin{enumerate}
\item $\Pi[[\alpha,\beta]] = [\Pi\alpha,\Pi\beta]_c \qquad \forall \quad \alpha, \beta$ .

\item
$[[\alpha, [[\beta,\gamma]] \ ]] = [[ \ [[\alpha,\beta]], \gamma]] + [[\beta, [[\alpha,\gamma]] \ ]]$ (Leibniz Identity).
\end{enumerate}
\end{corollary}
\begin{proof}
Combine Prop. \ref{prop1} and Prop. \ref{hagjacobi}.
\end{proof}

\section{Canonical bracket on $(p-1)$ forms}

Let $M$ be a $n$ dimensional manifold and $\Pi$ a $p$-vector $3 \leq p \leq n$. Ibanez et. al. \cite{Ibanez} have introduced a bracket, on $(p-1)$ forms, canonically associated to the Nambu-Poisson bracket $\{ \ , \ , \}$. We do not reproduce the proofs given by Ibanez and Hagiwara.
\begin{proposition}
\begin{equation}
[[\alpha,\beta]]_I = \mathcal{L}_{\Pi(\alpha)} \beta + (-1)^p \ (\Pi(d\alpha))\beta \quad \forall \quad \alpha, \beta
\end{equation}
is the unique bracket on $(p-1)$ forms such that
\begin{enumerate}

\item $\Pi$ is the anchor

\item
$
[df_1 \wedge \cdots \wedge df_{p-1}, dg_1 \wedge \cdots \wedge dg_{p-1}] = \sum_{i=1}^{p-1} dg_1 \wedge \cdots \wedge d\{f_1, \cdots , f_{p-1}, g_i \} \wedge \cdots dg_{p-1}
$
where $\{ \ , \ , \}$ is the Nambu-Poisson bracket on functions.
\end{enumerate}
\end{proposition}
\begin{proof}
See Ibanez et al Thm 3.6
\end{proof}

\begin{proposition}
$(\bigwedge^{p-1} T^\ast M, [[ \ , \ ]]_I, \Pi)$ is a Leibniz algebroid.
\end{proposition}
\begin{proof}
Thm. 3.7, Ibanez et. al. \cite{Ibanez}.
\end{proof}

\begin{corollary} \label{corIb}
\begin{enumerate}
\item $\Pi[[\alpha,\beta]]_I = [\Pi\alpha,\Pi\beta]_c \qquad \forall \quad \alpha, \beta$ if and only if 

\item the Leibniz identity holds for $[[ \ , \ ]]_I$.
\end{enumerate}
\end{corollary}
\begin{proof}
Note that the proof of Leibniz identity in Thm. 3.7, Ibanez. et. al. uses the bracket preserving conditions on $\Pi$ (Prop. 3.3 Ibanez et. al). Conversely our Prop. (\ref{prop1}) shows that $\mathit{1.} \implies \mathit{2.}$.
\end{proof}

\section{A new bracket on $(p-1)$ forms}

\begin{proposition} \label{diffbrack}
Let $[\ , \ ]$ and $(\ ,\ )$ be two Leibniz brackets on the same vector bundle. Let
\begin{equation}
\{A,B\} := [A,B] - (A,B), \qquad \forall A, B
\end{equation}
then $\{A,B\}$ is a  Leibniz bracket.
\end{proposition}
\begin{proof}
\textbf{Linearity} \\
\begin{eqnarray*}
\{A_1+A_2,B\} & = & [A_1+A_2,B] - (A_1+A_2,B)  \\
& = &  [A_1,B] + [A_2,B] - (A_1,B) - (A_2,B) \\
& = & [A_1,B]  - (A_1,B) + [A_2,B] - (A_2,B) \\
& = &  \{A_1,B\} + \{A_2,B\}  \qed
\end{eqnarray*}

\begin{eqnarray*}
\{c A,B\} & = & [c A, B] - (c A, B)  \\
& = &  c \ [A,B] - c \ (A,B) \\
& = & c \ \{A,B\}  \qed
\end{eqnarray*}

Similarly linearity in the second entry of the bracket follows.\\

\textbf{Leibniz identity} : Now $\{A,\{B,C\}\} = [A,[B,C]] - (A,(B,C))$, therefore  \\
\begin{eqnarray*}
& & \{A,\{B,C\}\} - \{\{A,B\},C\} - \{B,\{A,C\}\} \\
& = & [A,[B,C]] - (A,(B,C)) - [[A,B],C] + ((A,B),C) - [B,[A,C]] + (B,(A,C)) \\
& = & [A,[B,C]]  - [[A,B],C] - [B,[A,C]] - (A,(B,C)) + ((A,B),C)  + (B,(A,C))
\end{eqnarray*}

\textbf{Anchor Map} : For the difference bracket the anchor map is the difference of the two anchors.
\end{proof}

\begin{definition} \label{defdiffbrack}
Let $M$ be a manifold of dimension $n$ and $\Pi$ a Poisson $p$-vector. For $\alpha, \beta$ $(p-1)$-forms on $M$  define
\begin{equation} \label{newb}
\{\alpha,\beta\} = \imath_{\Pi(\beta)} d\alpha + (-1)^p (\Pi d\alpha)\beta.
\end{equation}
\end{definition}
\begin{remark}
Choosing the brackets $[\ , \ ] = [[\ , \ ]]_I$ and $(\ ,\ ) = [[\ , \ ]]$ in Prop. (\ref{diffbrack}), we obtain the bracket $\{ , \}$ of Def. \eqref{newb} as the difference bracket.
\end{remark}
\begin{proposition} \label{propnewb}
Consider the bracket as given in Def. (\ref{defdiffbrack}). We have
\begin{enumerate}
\item $\{\alpha,f\beta\} = f \{\alpha,\beta\} \quad \forall \quad (p-1)-forms \quad \alpha, \beta$, so that the zero map is the anchor for $\{ , \}$.

\item $\{ , \}$ is a Leibniz bracket and $(\bigwedge^{p-1} T^\ast M, [[ \ , \ ]]_I, \mathbf{0})$ is a Leibniz algebroid.
\end{enumerate}
\end{proposition}
\begin{proof}
\begin{enumerate}
\item
$\{\alpha,f\beta\}  =  \imath_{\Pi(f\beta)} d\alpha + (-1)^p (\Pi d\alpha)f\beta
 =  f \imath_{\Pi(\beta)} d\alpha + f (-1)^p (\Pi d\alpha)\beta = f \{\alpha,\beta\}$

\item Indeed the zero map preserves brackets. To verify the Leibniz identity for $\{ , \}$ we write it as the difference of two brackets $\{\alpha,\beta\}  =  (\mathcal{L}_{\Pi(\alpha)} \beta + (-1)^p (\Pi d\alpha)\beta) - (\mathcal{L}_{\Pi(\alpha)}\beta - \imath_{\Pi(\beta)} d\alpha)$. Thus the obstruction to the Leibniz identity for $\{ , \}$ is the difference of the Jacobi anomalies for the two brackets on RHS. By the preceding sections, these are the same, namely $\Pi(d(\Pi (d\alpha)\wedge \beta)) + (-1)^p (\mathcal{L}_{\Pi(\beta)}(\Pi(d\alpha)))$ (See Hagiwara proof of Thm 2.7 \cite{Hagiwara}. They correspond to the Nambu-Poisson tensor condition). Hence the difference is zero.
\end{enumerate}
\end{proof}

\begin{remark}
Condition (i) of Prop (\ref{propnewb}) is the zero anchor condition; it holds for other combinations of signs in the definition of the bracket Eq. (\ref{newb}). But $\{,\}$ gives a Leibniz algebroid.
\end{remark}

\section{Courant Brackets}

We are in a position to answer a question of K. Uchino \cite{Uchino}: Is the bracket preserving condition $\rho[[x,y]]_c = [\rho(x),\rho(y)]_c \quad x,y \ \in \ \Gamma(E)$ redundant in the definition of Courant algebroid ?
\begin{definition} \label{cour-def}
(Courant Algebroid \cite{Uchino,revgrab}) Let $E \rightarrow M$ be a vector bundle with a nondegenerate symmetric bilinear form $(,)$ and a skew-symmetric bracket $\rho[[x,y]]_c$ on $\Gamma(E)$and a bundle map $\rho: E \rightarrow TM$ such that
\end{definition}
\begin{enumerate}

\item $[[\ [[x,y]]_c,z]]_c + [[\ [[y,z]]_c,x]]_c + [[\ [[z,x]]_c,y]]_c   = DT(x,y,z) \quad \forall \ x,y,z \ \in \Gamma(E)$ \\

where $T(x,y,z) = \frac{1}{3}\left( ([[x,y]]_c,z) + ([[y,z]]_c,x) + ([[z,x]]_c,y) \right)$ and $D: \mathcal{C}^\infty(M) \rightarrow \Gamma(E)$ is $\frac{1}{2} \beta^{-1} \rho^\ast d$, $\beta$ being the isomorphism between $E$ and $E^\ast$ provided by the bilinear form $(,)$.

\item $\rho[[x,y]]_c = [\rho(x),\rho(y)]_c \quad x,y \quad \in \quad \Gamma(E)$.

\item $[[x,f y]]_c = f [[x,y]]_c + (\rho(x) f)y - (x,y)Df  \quad x,y \ \in \ \Gamma(E), \ \forall \ f \in \mathcal{C}^\infty(M)$.

\item $\rho \circ D = 0$.

\item $\rho(x) (y,z) = \big([[x,y]]_c + D(x,y)z \big) + \big(y,[[x,z]]_c + D(x,z) \big) \quad \forall \ x,y \ \in \ \Gamma(E)$.

\end{enumerate}

It was observed \cite{revgrab} that the Courant bracket is the antisymmetrization
\begin{equation} \label{cdtodb1}
[[x , y]]_c = \frac{1}{2}(x\circ y - y \circ x)
\end{equation}
for
\begin{equation} \label{cdtodb2}
x \circ y = [[x , y]]_c + \frac{1}{2}D(x,y) \quad (Dorfman \ Bracket)
\end{equation}
\begin{proposition}
$(E, \circ, \rho)$ is a Leibniz algebroid for the Dorfman bracket with anchor $\rho$
\end{proposition}
For Proof see Baraglia \cite{Baraglia}.

\begin{corollary}
$\rho[[x,y]]_c = [\rho(x),\rho(y)]_c$ i.e. the condition 2. in Def. \ref{cour-def} follows from equations \eqref{cdtodb1} and \eqref{cdtodb2}
\end{corollary}
\begin{proof}
By our prop 2.1, $\rho(x \circ y) = [\rho(x), \rho(y)]_c$. Therefore
\begin{eqnarray} \nonumber
\frac{1}{2} (\rho(x \circ y - y \circ x)) & = & \frac{1}{2} (\rho(x \circ y) - \rho(y \circ x)) \\ \nonumber
& = & \frac{1}{2} ([\rho(x),\rho(y)]_c - [\rho(y) \circ \rho(x)]_c) \\
& = & \frac{1}{2} ([\rho(x),\rho(y)]_c + [\rho(x) \circ \rho(y)]_c) = [\rho(x),\rho(y)]_c
\end{eqnarray}
since the bracket on tangent vectors is anti-symmetric.
\end{proof}

%

%
%


\section{Koszul/Fuchssteiner bracket on 1-forms}

The motivating example for brackets on higher degree forms is the Koszul bracket on 1-forms: given a bi-vector $\Pi$ define \cite{Kosmann}
\begin{definition}
$\{\alpha,\beta\} = \mathcal{L}_{\Pi(\alpha)} \beta = \mathcal{L}_{\Pi(\beta)} \alpha + d(\Pi(\alpha,\beta))$
\end{definition}
This bracket has been characterised by the following two conditions
\begin{enumerate}

\item $\{\alpha,f \beta\} = f \{\alpha,\beta\} + (\mathcal{L}_{\Pi(\alpha)} f) \beta$

\item $\{df, dg\} = - d\{f,g\}$

\end{enumerate}
where $\{f,g\}$ is the Poisson bracket corresponding to $\Pi$.

\begin{theorem} \label{kszthm} The following conditions are equivalent \cite{Kosmann}
\begin{enumerate}
\item The Jacobi identity holds for $\{,\}$.

\item $\Pi\{\alpha,\beta\} = \{\Pi\alpha,\Pi\beta\}$.

\item the Schouten bracket $[\Pi,\Pi]_S = 0$.
\end{enumerate}
\end{theorem}

\begin{remark}
For the brackets of Ibanez et. al. and of Hagiwara the equivalence of $\mathit{1.}$ and $\mathit{2.}$ of Thm. (\ref{kszthm}) holds as shown in Cor. (\ref{corIb}) and Cor. (\ref{cirHag}) respectively. Regarding $\mathit{3.}$ the vanishing of the Schouten bracket holds for even order Nambu-Poisson vectors $\Pi$. In general $\Pi$ need not be a Nambu-Poisson tensor even if $\mathit{3.}$ holds (Cor. \textit{III.8} of Ref. \cite{Ibanez2}).
\end{remark}

\section{Conclusion}

We have shown that the bracket-preserving property of the anhcor follows from the Jaobi identity. This redundancy was proved for Lie algebroids in a more laborious way; for Leibniz algebroids, it has been stated, but our simple proof seems worth presenting. The equivalence of the Jacobi identity and the bracket preserving property is established for the two brackets on $(p-1)$ forms on Nambu-Poisson manifolds. The redundancy of the bracket-preserving condition was conjectured by K. Uchino \cite{Uchino}: we have shown it here.

An outcome of our work is to further explore the role of the Leibniz bracket in mechanics. Our motivating example come from the Fuchssteiner bracket. In an early work, Fuchssteiner defined a bracket that extended the usual Poisson bracket of function to that of a bracket of closed one forms defined on a Poisson manifold. Using this bracket he was able to show that they form the symmetry algebra of several nonlinear evolution equations \cite{Fuchssteiner}. A natural question to ask now is, what is the role of the Leibniz brackets that are defined for $(p-1)$ forms ?

\end{document}